\newcommand{\ifConferenceVersion}{\iffalse}
\newcommand{\ifJournalVersion}{\iftrue}
\newcommand{\comment}[1]{}

\ifConferenceVersion
\documentclass[a4paper,english]{lipics-v2018}

\newcommand{\InConference}[1]{#1}
\newcommand{\InJournal}[1]{}
\usepackage{url}

\else
\documentclass[11pt,letter]{article}

\newcommand{\InConference}[1]{}
\newcommand{\InJournal}[1]{#1}

\usepackage{float}
\usepackage{url}
\usepackage{authblk}
\usepackage{fullpage}
\usepackage{graphicx}
\graphicspath{{./fig/}{.}}
\fi 

\usepackage[T1]{fontenc}
\usepackage{amsfonts}
\usepackage{amsmath}
\usepackage{amssymb}
\usepackage{booktabs}

\usepackage{array}
\usepackage[boxed,vlined,linesnumbered,noend]{algorithm2e}
\usepackage{algorithmic}
\usepackage{amsthm}

\newtheorem{fact}{Fact}
\InJournal{
\newtheorem{theorem}{Theorem} 
\newtheorem{lemma}[theorem]{Lemma}
\newtheorem{corollary}[theorem]{Corollary}

}

\newcommand{\Pro}[1]{\mathbf{Pr}\!\left[\,#1\,\right]}                      
\newcommand{\E}[1]{\mathbf{E}\!\left[\,#1\,\right]}                      
\newcommand{\BigO}[1]{\ensuremath{\mathcal{O}\!\left(#1\right)}}

\newcommand{\Decay}{\ensuremath{\mathsf{Decay}}\xspace}
\newcommand{\GD}{\ensuremath{\mathsf{Green\hbox{-}Decay}}\xspace}
\newcommand{\GB}{\ensuremath{\mathsf{GD\hbox{-}Broadcast}}\xspace}                      
\newcommand{\BB}{\ensuremath{\mathsf{Balls\hbox{-}into\hbox{-}Bins}}\xspace}
\newcommand{\GGB}{\ensuremath{\mathsf{BB\hbox{-}Broadcast}}\xspace}   
\newcommand{\SINGLE}{\ensuremath{\mathsf{Single}}\xspace}   
\newcommand{\etal}{{\it et~al.}}

\InConference{
	\title{Broadcast in Radio Networks:  Time vs. Energy Tradeoffs.}	
	\author{Marek Klonowski}{Wroclaw University of Science and Technology}{marek.klonowski@pwr.edu.pl}{}{}
	
	\author{Dominik Pajak}{Massachusetts Institute of Technology}{pajak@csail.mit.edu}{}{}
	
	\authorrunning{M. Klonowski and D. Pajak}
	
	\Copyright{Marek Klonowski and Dominik Pajak}
	
	\subjclass{\ccsdesc[100]{Theory of computation~Design and analysis of algorithms~Distributed algorithms}, \ccsdesc[100]{Networks~Network types~Ad hoc networks}}
	
	\keywords{radio networks; broadcast; energy efficiency; distributed algorithm;}
	
	\category{}
	
	\relatedversion{Full version available at: \url{https://arxiv.org/abs/1711.04149}.}
	
	\supplement{}
	
		\funding{The work of the first  author was supported by Polish National Science Center grant 2013/09/B/ST6/02258. The work of the second author was supported by Polish National Science Center grant 2015/17/B/ST6/01897.}
	
}

\begin{document}
	
	\InJournal{
\date{}
\title{
Broadcast in Radio Networks:  Time vs. Energy Tradeoffs.
\thanks{The work of the first  author was supported by Polish National Science Center grant 2013/09/B/ST6/02258. The work of the second author was supported by Polish National Science Center grant 2015/17/B/ST6/01897.}
}

\author[1]{Marek Klonowski}
\author[2]{Dominik Pajak}
\affil[1]{Wroclaw University of Science and Technology, Poland, E-mail: marek.klonowski@pwr.edu.pl}
\affil[2]{Massachusetts Institute of Technology, USA, E-mail: pajak@mit.edu}
}

\maketitle

\begin{abstract}
In wireless networks, consisting of battery-powered devices, energy is a costly resource and most of it is spent on transmitting messages. Broadcast is a problem where a message needs to be transmitted from one node to all other nodes of the network. We study algorithms that can work under limited energy measured as the maximum number of transmissions by a single station. The goal of the papnper is to study tradeoffs between time and energy complexity of broadcast problem in multi-hop radio networks.  We consider a model where the topology of the network is unknown and if two neighbors of a station are transmitting in the same discrete time slot, then the signals collide and the receiver cannot distinguish the collided signals from silence.


We observe that existing, time efficient, algorithms are not optimized with respect to energy expenditure. We then propose and analyse two new randomized energy-efficient algorithms.  Our first algorithm works in time $O((D+\varphi)\cdot n^{1/\varphi}\cdot \varphi)$ with high probability and uses $O(\varphi)$ energy per station for any $\varphi \leq \log n/(2\log\log n)$ for any graph with $n$ nodes and diameter $D$. Our second algorithm works in time $O((D+\log n)\log n)$ with high probability and uses $O(\log n/\log\log n)$ energy. 

We prove that our algorithms are almost time-optimal for given energy limits for graphs with constant diameters by constructing lower bound on time of $\Omega(n^{1/\varphi} \cdot \varphi)$. The lower bound shows also that any algorithm working in polylogaritmic time in $n$ for all graphs needs energy $\Omega(\log n/\log\log n)$.




\end{abstract}

\section{Introduction}
The problem of broadcast consists in delivering a single message from a source to all the nodes of a communication network. In multi-hop networks, neighboring stations can send messages to each other but when two neighbors of one station are sending at the same time, then these transmissions interfere and the messages are not delivered. Such a situation is in our model indistinguishable from silence (no collision detection). That is a station can locally distinguish only the case where exactly one neighbor transmits from all the other situations. The broadcast problem is fundamental in radio networks because it can be used to learn the topology of the network or as a subprocedure to other, more complex problems, like multi-message broadcast or gossiping~\cite{ChlebusGLP01}. It also allows to emulate single-hop (networks where any two stations can exchange messages) algorithms in multi-hop networks~\cite{Bar-YehudaGI91}. 

Two most important parameters of an algorithm in radio networks is the time complexity, measured as the number of steps necessary to complete the execution, and energy complexity, which is the maximum number of rounds in which a station is transmitting. For single-hop networks both time and energy complexity of broadcast algorithms has been well studied. However, in the more general case, for multi-hop networks, only time complexity was analyzed for a long time. Only very recently the first article appeared studying time and energy in multi-hop networks~\cite{pettie_energia}. We aim at showing algorithms that minimize the time as well as the energy in multi-hop networks. We want to also show a tradeoff between time and energy for broadcasting protocols. This will allow greater flexibility when designing algorithms by decreasing the maximum energy expenditure of a station at a cost of the runtime of the algorithm. Clearly, minimizing the energy cost can sometimes be a critical aspect of real-life systems as they are often composed of small, cheap, battery-powered devices whose batteries cannot be easily recharged or 
replaced. For such systems it may be reasonable to sacrifice time and 
save some energy of the stations.




\subsection{Model and Problem Statement}\label{model2}
In this paper we consider a radio network represented as an undirected, connected graph $G = (V,E)$,  where the nodes symbolize stations and the edges are bidirectional communication links between them. By $n$ we denote the number of nodes of the graph and by $D$ its diameter. We assume that the nodes are given value of $n$ and the energy limit $\varphi$. The stations do not know the topology of the network, are identical and do not have any labels.  Symmetry between the stations can be in our model broken 
as the stations have access to independent sources of random bits.

Time is divided into discrete \textit{rounds} and all the stations know the number of the current round. The model is synchronous and in each round each node either transmits a message or listens. Each station receives a packet from its neighbors only if it listens in a given  round and \textbf{exactly one} of its neighbors is transmitting. We say that a \SINGLE occurs in such a round. If zero or more than one neighbor of $v$ transmits, then $v$ receives no message.

The single-message broadcast problem is defined as follows. Initially some node, called \textit{originator} has a message and the goal is to deliver this message to all the nodes in the network. We assume that the nodes that have not received the message are not allowed to make any transmissions. In our setting broadcast can be also seen as a wakeup of the network. 
\InConference{\subparagraph{Energy metrics}}
\InJournal{\paragraph{Energy metrics}}
We define $e_v$, an \textit{energetic effort} of a station $v\in V$, as the number of rounds when $v$ transmitted. 
Note that both successful as well as unsuccessful (due to collisions) transmissions count. We will say that algorithm uses energy at most $E$ if $\max_{v\in V} e_{v} \leq E$. In other words we aim at limiting the energetic expenditure of \textbf{all} stations that are present in the network since we need all stations working. We will consider Monte Carlo algorithms using energy $E$ and time $T$ working with probability $p$ which will be understood that the algorithm always terminates after at most $T$ steps and each station uses at most $E$ energy and the broadcast is successful with probability at least $p$. The same definition was used for example in~\cite{Kopelowitz17,KlonowskiKZ12,ESA03}.

In some of the existing protocols the energy expenditure is a random variable and in order to compare it with our solutions we will analyse their energy metric defined as the expected maximum amount of energy spent by a station: $E[\max_{v\in V} e_{v}]$. Such a definition was used for example in~\cite{KardasKP13}. Note that in some articles in radio networks, also listening to the channel costs energy (see e.g.,~\cite{BenderKPY16,pettie_energia,ChangKPWZ17}). However in this work we assume that the transmission of a packet costs much more energy than reception of it and therefore we aim at minimizing only the number of transmissions by each station. That is, we assume that energy cost of listening is negligible comparing to transmitting. Such assumption  is particularly justified for systems with large distances between devices.



\InConference{\subparagraph{Notation}}
\InJournal{\paragraph{Notation}}
For any $k$, set of integers $\{1,2,\dots,k\}$ is denoted by $[k]$. By $\log x$ we denote logarithm at base $2$ of $x$ and by $\ln x$  the natural logarithm.

\subsection{Our results}
In this paper we present two energy-efficient broadcast algorithms and a lower bound. 
We first show an universal algorithm \GGB in which the available energy is a parameter. This allows us to obtain a complete tradeoff between time and energy for any fixed energy between constant and $\log{n}/\log\log{n}$.  The second algorithm \GB is a modification of the broadcast algorithm from~\cite{Bar-YehudaGI91}. The goal of the second algorithm is to perform broadcast in the same (almost-optimal) time as in~\cite{Bar-YehudaGI91} but reduce its energy complexity. The obtained algorithm has energy complexity $O(\log{n}/\log\log n)$ which we will later show to be the minimum energy complexity for any algorithm with time polylogarithmic in $n$. We also present a lower bound which shows that our algorithms are almost time-optimal for graphs with constant diameter. 
Finally we show a lower bound of $\Omega(n^{3/2})$ for algorithms using at most two transmissions. 
\begin{table}\centering
	\begin{tabular}{lllll}\toprule[1pt]
		\textbf{\textsc{Time}} & \textbf{\textsc{Energy}}  & \textbf{\textsc{Prob.}} & \textbf{\textsc{Remarks}} & \textbf{\textsc{Ref.}} \\\midrule[1pt]
		$\BigO{(D + \log n)\log n}$  & $\BigO{\frac{\log n}{\log\log n}}$ & $1-1/n$ &&  Thm~\ref{thm:gd_broadcast} \\\hline
		$\BigO{\left(D + \frac{\log n}{\alpha\log\log n}\right)  \frac{\log^{1+\alpha}n}{(\alpha-1)\log\log n}}$ & $\frac{\log (n/\epsilon)}{\alpha\log\log n} + o(1)$ & $1-\epsilon$ & $\alpha > 1$ & Thm~\ref{thm:bb_broadcast}  \\\hline
		$\Omega\left(\frac{\log^{1+\alpha}n}{\alpha \log\log n}\right)$ & $\frac{\log n}{\alpha\log\log n}$ & const. & constant $D$ & Corr~\ref{cor:lower}\\\hline		
		$\BigO{(D + \varphi)\cdot n^{1/\varphi}\cdot \varphi}$ & $\left\lceil\varphi\cdot \left(1 + \frac{\log{2/\epsilon}}{\log{n}}\right)\right\rceil $  & $1-\epsilon$ & $\varphi \leq \frac{\log{n}}{2\log\log{n}} $  & Thm~\ref{thm:bb_broadcast} \\\hline
		$\Omega(n^{1/\varphi}\cdot \varphi)$ & $\varphi$  & const. & constant $D$ &
		Corr~\ref{cor:lower}\\\hline
		$\Omega(n^{3/2})$ & $2$  & $1-1/n$ & $D = \Theta(n)$ &
		Thm~\ref{thm:energy2}\\
		\bottomrule[1pt]
	\end{tabular}
	\caption{Summary of our results. For example by setting $\varphi = \frac{\sqrt{\log n}}{\log\log n}$ shows that there exists an algorithm with energy $\BigO{\sqrt{\log n}}$ and time $\widetilde{\mathcal{O}}(D\cdot 2^{\sqrt{\log n}})$ whereas time $\Omega(2^{\sqrt{\log n}})$ is required for such energy for some graphs with constant diameter. Similarly for energy $\BigO{\log\log n}$, our algorithm uses $\widetilde{\mathcal{O}}(D\cdot n^{1/\log\log n})$ time steps.}
\end{table}

%

\subsection{Related work}\label{related2}


Energy aspect of broadcast in known and unknown networks has been investigated by Kantor and Peleg in~\cite{KantorP16} using the same energy-efficiency metric as in our paper. For the model with unknown topology they presented a protocol that uses $k$ energy and is completed  in $\BigO{ (D+\min\{D\cdot k, \log n \}) \cdot n^{1/(k-1)} \log n }$ rounds  for $k>1$ and $\BigO{D\cdot n^2 \log n}$ for $k=1$ w.h.p. Moreover they demonstrated $\Theta(\log n)$-shot broadcasting protocol that terminates in $\BigO{  D\log n  + \log ^2 n }$ rounds.

Since authors of~\cite{KantorP16} use the same model, the results can be compared with ours. First, let us note that their protocols are based on a different constructions. For time $\BigO{  D\log n  + \log ^2 n }$ our algorithm \GD uses $\log\log n$ times less energy. We also show our energy bound is tight and energy $\Theta(\log n/\log\log n)$ is the always required for any algorithm has time polylogarithmic in $n$. We can also compare our results that use any energy. Algorithm of Kantor and Peleg works with probability $1- n^{-1/(k-1)}$. If we use the same probability of success then our algorithm \GGB will have time complexity $\BigO{(D + k)n^{1/(k-1)}k}$ which is by a factor of $\log n / k$ faster. 



In~\cite{BerenbrinkCH09}  Berenbrink~\etal~presented $1$-shot broadcast protocol that needs  $\BigO{\log n}$ steps in a random and unknown network w.h.p. For general networks with known diameter $D$, they presented a randomized broadcast algorithm with optimal time $\BigO{D \log (n/D) + \log^2 n  }$ with  expected $\BigO{\log^2 n/\log (n/D) }$ transmissions per node.

To the best of our knowledge, the rest of related work about broadcast  uses a substantially different metrics  of energy consumption or does not take  into account energetic aspects at all. 

Recently Chang~\etal~\cite{pettie_energia} presented various results about energy-efficient broadcast. Their assumptions differs however from ours in two aspects. In~\cite{pettie_energia} the authors assume that the stations which have not received the message are allowed to make transmissions. Secondly the energy complexity in~\cite{pettie_energia} includes the number of times a station listens to the channel. Because of that the minimal energy complexity can be shown to be $\Omega(\log\Delta \log n)$ for graphs with maximum degree $\Delta$. In contrast, in our setting, where only the number of transmissions. The authors in~\cite{pettie_energia} showed an algorithm with time complexity $\BigO{D^{1+\epsilon} \log^{O(1/\epsilon)} n}$ and energy $\BigO{\log^{O(1/\epsilon)} n}$ for the model without collision detection.

The first randomized broadcast protocol (without considering energy-efficiency aspects)   presented by Bar-Yehuda~\etal~\cite{Bar-YehudaGI92} for the model without collision detection works in any multi-hop network in time $\BigO{D\log{n} + \log^2 n}$ with high probability. It is based on \Decay procedure from~\cite{Bar-YehudaGI91} extensively used also in many other papers. Improved protocols with expected time  $\BigO{D\log(n/D) + \log^2 n}$ have been independently proposed by Czumaj and Rytter~\cite{CzumajR06} and Kowalski and Pelc~\cite{KowalskiP05}.  Those results are  optimal due to the lower bounds $\Omega(\log^2 n)$ shown by Alon~\etal~\cite{alon1991lower} and $\Omega(D\log(n/D))$ by Kushilevitz and Mansour~\cite{KushilevitzM98}. Recently Haeupler and Wajc~\cite{HaeuplerW16} proved that if stations are allowed to make transmissions before receiving the message then broadcast can be completed in time $\BigO{D \cdot \frac{\log{n}\log\log n}{\log{D}} + \log^{O(1)}n }$ which was improved by
to $\BigO{D \cdot \frac{\log{n}}{\log{D}} + \log^{O(1)}n }$ Czumaj and Davies~\cite{CzumajD17} . 

In the model with known topology Gasieniec~\etal~\cite{GasieniecPX07} showed a randomized algorithm with time $\BigO{D + \log^2 n}$ and Kowalski and Pelc~\cite{KowalskiP07DC} showed a deterministic algorithm with the same complexity. The algorithms are time-optimal since the bound $\Omega(\log^2 n)$ by Alon~\etal~\cite{alon1991lower} holds also for the known topology. In some papers, broadcast in specific graph classes, like line~\cite{DiksKKP02} or planar graphs~\cite{ElkinK07, GasieniecPX07} was considered.   Kantor and Peleg~\cite{KantorP16} also considered energy efficiency of broadcast with known topology. They demonstrated a protocol executed in   $D + \BigO{k n^{1/2k}\log^{2+1/k} n }$ rounds for a network of $n$ nodes with diameter $D$ and known topology.  For such setting they also presented  a $\Omega(D+k\cdot (n-D)^{1/2k})$~ lower bound  for the number of rounds. Results regarding the model with known topology  can be seen as an extension of results from an earlier results from~\cite{GasieniecKKPS08}~.

\section{Energy-efficient broadcast algorithms}\label{green2}
In this section we present two new algorithms. The first, \GGB ("Balls-into-Bins Broadcast") can work with arbitrarily small (also constant) available energy. Of course smaller energy leads to higher running time. The second algorithm \GB ("Green Decay Broadcast") is built based on classical Broadcast by Bar-Yehuda et al.~\cite{Bar-YehudaGI92}. Our \GB algorithm works in the same asymptotic time as the original one, but has reduced energy complexity $\BigO{\log n/\log\log n}$. We will later show that any algorithm that works in polylogarithmic time on graphs with constant diameter needs at least this much energy.  

\subsection{Balls into Bins}
Here we introduce a subprocedure $\BB(k)$ that will be used in both our algorithms. In this subprocedure each participating station transmitts in one, randomly chosen out of $k$ time slots. Balls into Bins
(called also Random Mapping) is a classical 
concept 
in probability theory with applications in load balancing where the studied value is usually the maximum number of balls in any bin (e.g. \cite{kolchin1986random}). In our setting, balls correspond to transmissions by the stations and bins are the time slots, hence we are interested whether some bin contains exactly one ball as this corresponds to a successful transmission. 
\begin{algorithm}[h]
	Let $t$ be the first time slot of the subprocedure\\
	Choose a random number $i \in [0,1,..,k-1]$\\
	\texttt{Transmit} in slot $t+i$
	\caption{$\BB(k)$}
\end{algorithm}

In the following lemma we bound the probability that in a procedure $\BB(k)$ there is a slot that is  chosen by exactly one station. 
\begin{lemma}
	\label{lem:bb}
	If $m$ neighbors of any fixed node $v$ are performing procedure $\BB(k)$  with $k = 24\lceil n^{1/\varphi} \rceil  + 1$ if $1 \leq \varphi < \frac{\log n}{\log\log n}$ and if $1 \leq m \leq 12/\varphi\cdot  n^{1/\varphi}\ln n$ then $v$ receives the message with probability at least $1- \frac{1}{2n^{1/\varphi}}$ for sufficiently large $n$.
	
\end{lemma}
\begin{proof}
	
	The probability that $v$ receives the message is equal to the probability that when throwing $m$ balls into $k$ bins (independently and uniformly at random) at least one bin contains a single ball. 
	
	The case $m=1$ is trivial. Consider case $1 < m < 12 \lceil n^{1 / (2\varphi)} \rceil$. Take the two last balls and observe that the probability that (at the moment when each of the balls is thrown) either of them lands in one of the bins that have not been already occupied is at least 
	\[
	1-\left(\frac{12 \lceil n^{1/(2\varphi)} \rceil}{k}\right)^2.
	\]
	With probability $1/k$ the last two balls collide hence with probability at least 
	\[
	1-\left(\frac{12 \left\lceil n^{1/(2\varphi)}\right\rceil}{k}\right)^2 - \frac{1}{k} \geq 1 - \frac{1}{2n^{1/\varphi}},
	\]
	at least one of the two last balls ends up as the only ball in one of the bins. For $12/\varphi \cdot n^{1/\varphi} \ln n \geq 	m \geq 12 \lceil n^{1/(2\varphi)} \rceil$ we use ~\cite[Lemma 4]{KlonowskiKZ12} and obtain:
	\begin{align*}
	\Pro{X = 0} &\leq \exp\left(-\frac{m}{2}	\left(1-\frac{1}{k}\right)^{2m-2}\right) = \exp\left(-\frac{m}{2}	\left(\left(1-\frac{1}{k}\right)^{k-1}\right)^{\frac{2m-2}{k-1}}\right)\\
	& \leq \exp\left(-\frac{m}{2}\left(\exp\left(-\frac{2m-2}{k-1}\right)\right)\right) = f(m,k).
	\end{align*}
	The derivative of $f(n,k)$ with respect to $m$ is equal to:
	
	\[
	\exp\left(-\frac{m}{2}\left(\exp\left(-\frac{2m-2}{k-1}\right)\right) - \left(\exp\left(-\frac{2m-2}{k-1}\right)\right)\right)\cdot \left(\frac{m}{k-1} - \frac{1}{2}\right).
	\]
	The function has a single minimum for $2m = k - 1$ and because the derivative is negative for $2m < k - 1$, the maximum value is attained in one of the endpoints of the considered interval. Knowing that $\varphi < \frac{\log n}{\log\log n}$, we get 
	\[
	f(12 \lceil n^{1/(2\varphi)} \rceil,24\lceil n^{1/\varphi} \rceil + 1) \leq \exp\left(-6n^{1/\varphi}\right) \leq \frac{1}{2n^{1/\varphi}}.
	\]
	\begin{align*}
	f(12 \lceil1/\varphi \cdot  n^{1/\varphi} \ln n\rceil,24\lceil n^{1/\varphi} \rceil + 1) &\leq  \exp\left(-6/\varphi \cdot n^{1/\varphi} \ln n\exp\left(-\frac{24\lceil1/\varphi \cdot n^{1/\varphi}\ln n\rceil}{24\lceil n^{1/\varphi} \rceil}\right)\right) \\&\leq \frac{1}{2n^{1/\varphi}},
	\end{align*}

\end{proof} 

This lemma cannot be easily improved by more than a constant factor, since in Balls into Bins model if we want each bin to contain at least $b$ balls, the total expected needed number of needed balls is close to $k \log k  + (b-1) \cdot k \log\log k$~\cite{Dixie}. And the concentration bound is very strong -- having $k\log k  + k \log\log k + a k$ balls, the probability that each bin contains at least two balls is for large $a$ close to $e^{-e^{-a}}$~\cite{analcomb}.  
\subsection{Balls into Bins Broadcast}\label{sec:bb}
The goal of this section is to design an algorithm that uses at most $\BigO{\varphi}$ transmissions and has time complexity roughly $\BigO{D \cdot n^{1/\varphi} \cdot \varphi}$ for any $\varphi > 0$. The idea of the algorithm is as follows. We consider the algorithm from the perspective of a fixed node $u$ at the first step when at least one of its neighbors $v$ has the message. The goal is to deliver the message to $u$ with high probability in time $\BigO{n^{1/\varphi}}$. Each participating station (neighbor of $u$ that has the message) chooses a phase being a number chosen according to geometric distribution with parameter $\varphi/ n^{1/\varphi}$. We will show that regardless of how many participants are there, some value will be chosen by at most $\BigO{n^{1/\varphi}/\varphi \cdot \ln n}$ stations. Each phase is a $\BB(n^{1/\varphi})$ procedure. We already know that if the number of participants is $\BigO{n^{1/\varphi}/\varphi \cdot \ln n}$, then the failure probability of such procedure (i.e. no \SINGLE 
transmission occurs ) is at most $n^{-1/\varphi}$. Then, 
repeating it $2\varphi$ times will give us high probability. This intuition is further specified in the following pseudocode and formalized in the next two lemmas.

\begin{algorithm}[]\label{alg1}
	$a \gets \left\lceil \frac{\varphi \log{n}}{\log{n} - \varphi \log{\varphi}} \right\rceil $\\
	$k \gets 24 \lceil n^{1/\varphi}\rceil + 1$ \\
	$t_{ph} \gets a k$ \\
	Wait until receiving the message;\\	
	\Repeat{$ \left\lceil\varphi\cdot \left(1 +\frac{\log{\frac{2}{\epsilon}}}{\log{n}}\right)\right\rceil $ times}{
		Wait until $(Time \texttt{ mod } t_{ph}) = 0$\\
		Choose a number $x \sim \mathsf{Geo(\varphi/n^{1/\varphi})}$\\
		Skip $(\min\{x, a\}-1) \cdot k$ rounds\\
		$\BB(k)$
	}
	\caption{$\GGB(\varphi,\epsilon)$}	
\end{algorithm}
The following lemma shows that, regardless of how many stations execute line $7$ of algorithm \GGB in parallel, some number is chosen by at most $12/\varphi \cdot n^{1/\varphi} \ln n $ stations. 
\begin{lemma}
	\label{lem:dist}
	For any $1 \leq \varphi < \frac{\log{n}}{\log\log{n}}$ and for $a  = \left\lceil \frac{\varphi \log{n}}{\log{n} - \varphi \log{\varphi}} \right\rceil$, if $1\leq \hat{n} \leq n$ random variables $X_1,X_2,\dots,X_{\hat{n}}$ are chosen from $\mathsf{Geo(\varphi/(n^{1/\varphi}))}$ then with probability at least $1-2/n^{2}$ among $Y= \{Y_{i}\}_{i=1}^{\hat{n}}$  ($Y_{i} = \min\{X_i,a\}$) there is a number $y\in \{1,2,\dots,a\}$ chosen at least once and at most $12/\varphi \cdot  n^{1/\varphi}  \ln n $ times by $Y$ variables.
	pp\end{lemma}
\begin{proof}
	
	If $\hat{n} \leq 12/\varphi \cdot n^{1/\varphi}\ln{n}$ then the statement trivially follows. Consider the opposite case. We have a set of $\hat{n}$ independent identically distributed random variables $X_1,\dots,X_{\hat{n}} \sim \mathsf{Geo(\varphi/n^{1/\varphi})}$. Observe that since $\varphi < \frac{\log n}{\log\log n}$ then $\varphi \log \varphi < \log n$ and if we denote $b = \left(\frac{1}{\varphi}- \frac{\log{\varphi}}{\log{n}}\right)$ then $b > 0$. Observe that $a = \lceil 1/b \rceil$ (line $1$ of pseudocode). We have for each $j \in \{1,2,\dots,\hat{n}\}$:
	\[
	\Pro{X_j > i} = \varphi^{i}n^{-i / \varphi} = n^{-i\left(\frac{1}{\varphi} - \frac{\log{\varphi}}{\log{n}}\right)} = n^{-i\cdot b},\quad\text{for any }  i = 1,2,\dots,a.
	\]
	
	\noindent
	Thus $\E{|j : X_{j} > i|} = \hat{n} \cdot n^{-i\cdot b}$. We know that $\hat{n} > 12/\varphi \cdot  n^{1/\varphi}\ln{n}$ and $n^b = n^{1/\varphi}/\varphi$ thus $\hat{n} > 4 n^b \ln{n}$. Take the smallest $i^*$ such that $\hat{n} / n^{i^* \cdot  b} <  4 n^b \ln{n}$. Since $b > 0$, such $i^*$ exists and $i^* \geq 1$ and moreover since $\hat{n} \leq n$ then $i^* \leq 1/b \leq a$. Using the minimality of $i^*$ we can write  $\hat{n} = n^{i^* \cdot b} \cdot r$, where $4n^{b}\ln{n} \geq r \geq 4\ln{n}$
	We define the following variables 
	\[
	Z_j = \left\{\begin{array}{lr}
	1& \quad \text{if } X_j > i^*,\\
	0& \quad \text{otherwise.}
	\end{array}\right.
	\]
	If $Z = \sum_{i=1}^{\hat{n}} Z_i$ then $\E{Z} = r \geq 4 \ln{n}$ and by Chernoff bound:
	\[
	\Pro{Z = 0} \leq \Pro{Z  \leq (1-1)\E{Z}} \leq e^{-\frac{\E{Z}}{2}} \leq \frac{1}{n^2},
	\]
	\[
	\Pro{Z > 12/\varphi \cdot n^{1/\varphi}\ln{n}} = \Pro{Z > 12n^{b}\ln{n}} \leq \Pro{Z > 3\E{Z}} \leq e^{-\frac{2 \E{Z}}{3}}\leq \frac{1}{n^2}.
	\]
	Thus with probability at least $2/n^{-2}$ between $1$ and $12/\varphi\cdot n^{1/\varphi}\ln{n}$ values of $X$ variables are at least $i^*$. Hence there is some value of $Y$ variables chosen at least once and at most $12/\varphi \cdot n^{1/\varphi}\ln{n}$ times.
	
	%
\end{proof}	

For all nodes we define some tree of shortest paths connecting the originator $u$ to all the nodes of the graph. For any node $v$ we denote by $\mathcal{P}_v$ the shortest path from $u$ to $v$ in this tree and by $p(v)$ the second last node on this path. To analyze the complexity of the algorithm we want to bound the speed at which the message progresses along path $\mathcal{P}_v$. We need to first show that with high probability if $p(v)$ receives the message then also $v$ receives it. Denote by $T(p(v))$ the time when $p(v)$ receives the message. Then consider the next $2\varphi$ complete phases that start after $T(p(v))$. We say that a phase is successful for $v$ if it results in a \SINGLE and delivers the message to $v$. Note that $v$ might receive the message before $T(p(v))$ but we say that the successful phase for $v$ is the first after $T(p(v))$.
\begin{lemma}
	\label{lem:ggb}
	For any $n\geq 4$ and $1 \leq \varphi < \frac{\log n}{\log\log n}$ in $\GGB(\varphi,\epsilon)$:
	\begin{enumerate}
		\item each phase is successful with probability at least $1-\frac{1}{2n^{1/\varphi}} - \frac{2}{n^2},$
		\item with probability at least $1- \epsilon/2$ some of the $\left\lceil\varphi\cdot \left(1 +\frac{\log{\frac{2}{\epsilon}}}{\log{n}}\right)\right\rceil $ phases is successful at each node of path $\mathcal{P}_v$ for every $v$.
	\end{enumerate}
\end{lemma}
\begin{proof}
	In any phase, by Lemma~\ref{lem:dist} with probability at least $1-2/n^{-2}$, some procedure $\BB$ is executed by at most $12/\varphi \cdot n^{1/\varphi}\log{n}$ stations. In such a case by Lemma~\ref{lem:bb}, procedure $\BB$ obtains a \SINGLE with probability at least $\frac{1}{2 n^{1/\varphi}}$. Hence, with probability at least $1 - \frac{1}{2 n^{-1/\varphi}} - \frac{2}{n^2}$ the phase is successful (results in a \SINGLE in one of its time slots). Using independence, the probability that $\left\lceil\varphi\cdot \left(1 +\frac{\log{\frac{2}{\epsilon}}}{\log{n}}\right)\right\rceil $ phases are unsuccessful is at most:
	\[\left(\frac{1}{2n^{1/\varphi}}+\frac{2}{n^{2}}\right)^{\left\lceil\varphi\left(1 + \frac{\log{\frac{2}{\epsilon}}}{\log{n}}\right)\right\rceil }\leq  \frac{\epsilon/2}{n}\left(\frac{1}{2} + \frac{2}{n^{2 - 1/\varphi}}\right)^{\left\lceil\varphi\left(1 + \frac{\log{\frac{2}{\epsilon}}}{\log{n}}\right)\right\rceil} \leq \frac{\epsilon/2}{n},\]
	where the last inequality holds because $n^{2-1/\varphi} \geq 4$ thus $1/2 + 2/n^{2-1/\varphi} \leq 1$.
	Hence for any node $v$, the probability that all of the $\left\lceil\varphi \left(1 + \frac{\log{\frac{2}{\epsilon}}}{\log{n}}\right)\right\rceil$ phases executed by $p(v)$ are unsuccessful is at most $\epsilon/2 \cdot n^{-1}$.  Taking union bound over all $n$ stations, we get that with probability at least $\epsilon/2$, some phase is successful at every node. Now observe that since paths $\mathcal{P}$ form a tree this is sufficient to prove $2$.
\end{proof}

The previous lemma shows that with probability at least $1-\epsilon/2$ each station eventually receives the message. This does not prove that $v$ receives the message \textit{from} $p(v)$ but only that during the $\left\lceil\varphi\cdot \left(1 +\frac{\log{\frac{2}{\epsilon}}}{\log{n}}\right)\right\rceil$ phases executed by $p(v)$, node $v$ receives the message (possibly from a different neighbor). But with probability $1-\epsilon/2$ the time until the message reaches $v$ is at most $2\varphi$ phases after $p(v)$ receives the message. Using this, we could bound the total number of phases until each station receives the message by $\BigO{D \cdot \varphi}$. We want however a  better bound of $\BigO{D + \varphi}$ using the fact that on average only a constant number of phases is sufficient to deliver the message from one node to its neighbor. 

\begin{theorem}
	\label{thm:bb_broadcast}
	For any $1 \leq \varphi < \frac{\log n}{\log\log n}$ and $\epsilon > 2n^{-3}$, if $n \geq 4$, then Algorithm $\GGB(\varphi,\epsilon)$ completes broadcast 
	\begin{enumerate}
		\item in time $\BigO{\left(D+\varphi\right)\cdot n^{1/\varphi}\cdot\frac{\varphi \log{n}}{\log{n} - \varphi\log{\varphi}}}$,\label{itm:time}
		\item using at most $\left\lceil\varphi \left(1 + \frac{\log{(2/\epsilon)}}{\log{n}}\right)\right\rceil$ energy per station,\label{itm:energy}
		\item with probability at least $1-\epsilon$.
	\end{enumerate}
\end{theorem}
\begin{proof}
	Claim \ref{itm:energy} follows directly from the construction of the algorithm. Take any vertex $v$ and consider path $\mathcal{P}_v = (u,v_1,v_2,\dots,v_{D_v-1},v)$ of length $D_v \leq D$ from the originator $u$ to $v$. Let us introduce random variables $X^{(v)}_i$ as the number of phases between the reception of the message by $i$-th node on path $\mathcal{P}$ and the first successful phase for $(i+1)$-st node on the path. If none of the $\left\lceil\varphi\left(1 + \frac{\log{\frac{2}{\epsilon}}}{\log{n}}\right)\right\rceil$ phases for $(i+1)$-st node on the path are successful we set $X^{(v)}_i = \infty$. Observe that the number of phases until $v$ receives the message is upper bounded by $\sum_{i=1}^{D_v} X^{(v)}_i$. 
	
	We know by Lemma~\ref{lem:ggb} that each phase is successful independently with probability at least $1-\frac{1}{2n^{1/\varphi}} - \frac{2}{n^2} \geq 1-\frac{1}{n^{1/\varphi}}$. Moreover $X^{(v)}_i \leq \left\lceil\varphi\left(1 + \frac{\log{\frac{1}{2\epsilon}}}{\log{n}}\right)\right\rceil$ for all $v$ and $i$ with probability at least $1 - \epsilon/2$. Observe that conditioned on this event, variable $X^{(v)} = \sum_{i=0}^{D_v} X^{(v)}_i$ is stochastically dominated by a sum of $D_v$ geometric variables $Y^{(v)} = \sum_{i=0}^{D_v} Y^{(v)}_i$ with success probability $1-1/n^{1/\varphi}$.  We have that $\E{Y^{(v)}} = D_v/(1-1/n^{1/\varphi})$.  We can use concentration bound for the sum of geometric variables~\cite[Theorem 2.3]{janson2017tail}. Let $\lambda =  6\left(1 + \frac{\varphi}{\E{Y^{(v)}}}\right) $:
	\[
	\Pro{\sum_{i=0}^{D_v} Y^{(v)}_i \geq  \lambda\E{Y^{(v)}}} \leq \frac{1}{\lambda} \cdot \left(\frac{1}{n^{1/\varphi}}\right)^{(\lambda - 1 -\ln{\lambda})\E{Y^{(v)}}} 
	\]
	since $\lambda > 6$, $\ln\lambda < \lambda/3$ and:
	\[
	\Pro{\sum_{i=0}^{D_v} Y^{(v)}_i \geq  \lambda\E{Y^{(v)}}} \leq \frac{1}{n^{4}}.
	\]
	Observe moreover that $\lambda \E{Y^{(v)}} = \BigO{D + \varphi}$.
	By taking union bound in both cases over all vertices $v$ we get that with probability at least $1- 1/n^3$,  the number of phases until each node receives the message is $\BigO{D + \varphi}$, conditioned on the fact that at least one phase is successful for each node. Since we know that the latter event takes place with probability at least $1- \epsilon/2$ then with probability at least $1-\epsilon$ all the nodes receive the message and the total number of phases is $\BigO{D + \varphi}$. Observe that the stations terminate the algorithm after additional time at most $\left\lceil\varphi\cdot \left(1 +\frac{\log{\frac{2}{\epsilon}}}{\log{n}}\right)\right\rceil \cdot  t_{ph}$. Thus, since $t_{ph} \in \BigO{n^{1/\varphi}\cdot\frac{\varphi \log{n}}{\log{n} - \varphi\log{\varphi}}}$, we obtain the desired result.
\end{proof}	
Observe that if $\varphi \leq \frac{\log{n}}{2 \log\log{n}}$ 
then $\frac{\varphi \log{n}}{\log{n} - \varphi\log{\varphi}} = \Theta(\varphi)$ and the complexity of the algorithm becomes $\BigO{(D+\varphi)n^{1/\varphi} \cdot \varphi}$.
\subsection{Green-Decay}
Algorithm \GGB can operate under a wide range of energy limits, however for the optimal energy $\log n/(\log\log n + 1)$ 
guarantees only 
$\BigO{\left(D +\frac{\log n}{\log\log n}\right) \cdot \frac{\log^2 n}{\log\log\log n}}$, which is slower than the algorithms from the literature. In this section we want to develop an algorithm that is less universal (i.e. does not offer parametrized energy vs. time trade-off) but achieves an almost-optimal time $\BigO{(D+\log n)\log n}$ using optimal energy $\BigO{\log n/\log\log n}$.

We will first present \GD which is a simple modification of classical procedure \Decay introduced in~\cite{Bar-YehudaGI91}. It will serve as a subprocedure to our energy-efficient algorithm \GB. In original \Decay, each station 
transmits 
for a number of rounds being a geometric random variable. We note that instead of broadcasting in each round of the procedure it is sufficient to broadcast only in the last one. With this we save energy whilst the probability of success remains the same. 
\begin{algorithm}[]
	\texttt{Transmit}\;
	\Repeat{$x = 1$ but at most $k$ times}{
		$x \gets$ $0$ or $1$ with equal probability\;
		\If{$x = 1$}{
			\texttt{Transmit}\;
		}
		
	}
	\caption{$\GD(k)$}
\end{algorithm}

The following theorem is analogous to~\cite[Theorem 1]{Bar-YehudaGI91}. By inspection of the proof from~\cite{Bar-YehudaGI91} we observe that exactly the same proof works for $\GD$. Furthermore it is easy to see that 
the modified procedure 
uses only constant energy. That is a station executing modified procedure transmits only once. 
\begin{theorem}
	\label{thm:gd}
	If $n$ neighbors of station $v$ execute procedure $\GD(k)$ then the probability $\Pro{k,n}$ that $v$ receives the message satisfies:	
	\begin{enumerate}
		\item $\Pro{\infty,n} = \frac23,$
		\item $\Pro{k,n} > \frac{1}{2}$, for $k\geq 2\log n.$
	\end{enumerate}
\end{theorem}
%

The high-level idea of \GB is as follows. In $\mathsf{Broadcast}$ algorithm from~\cite{Bar-YehudaGI92} each station participates in $\Theta(\log n)$ \Decay procedures and the expected maximum energy is $\Theta(\log{n}\log\log{n})$. By simply replacing it with \GD (that takes always at most a constant energy per participant) we can reduce the energy complexity from $\log n \log\log n$ to $\log n$. In order to reduce the energy complexity further we observe that a station does not necessarily need to participate in all the $\log n$ procedures \Decay. If some number $x$ of neighbors of $v$ have the message and want to transmit it to $v$ it is sufficient that in at least a constant fraction
of the $\Theta(\log n)$ procedures \Decay at least one among the $x$ stations participate. In our algorithm each station participates in $\Theta(\log n/\log\log n)$ procedures \Decay chosen at random. If $x$ is sufficiently large, at least a constant fraction of procedures \Decay will have at least one participant and the algorithm will work 
correctly with high probability. On the other hand if $x$ is small we can use procedure $\BB$ which gives a probability of success of order $1-1/\log n$, for $\varphi = \log{n}/\log\log{n}$. Hence for small $x$, $\BigO{\log n/\log\log n}$ procedures $\BB$ is sufficient to obtain the high probability
of successful transmission. 
Our algorithm combines \GD and \BB to cover both cases of small and large number of neighbors trying to deliver the message. One more difficulty we need to overcome is that the number of participating (i.e., holding the message) neighbors of $v$ might increase over time. 
\begin{algorithm}[t]
	$ll \gets \lceil \log\log n \rceil$\\
	$k \gets 24\lceil \log n\rceil + 1$\\
	$state \gets \texttt{new}$\\
	Wait until receiving the message;\\
	\Repeat{$2\lceil\log{n}\rceil + 2$ times}{
		Wait until $(Time \texttt{ mod } 3 \cdot k) = 0$\\
		$phase \gets Time/(3\cdot k) \texttt{ mod } ll$\\
		\If{$phase = 0$ and $state = $ {\upshape \texttt{new}}}{
			$state \gets \texttt{normal}$ \\
		}{\textbf{endif}}\\
		\If{$state =$ {\upshape \texttt{new}}}{
			$\BB(k)$ \\
			Skip $k$ rounds \\ 
		}
		\ElseIf{$phase = 0$}{
			Skip $k$ rounds \\
			$\BB(k)$\\
		}{\textbf{endif}}
		
		\If{$phase = 0$ or $state =$  {\upshape \texttt{new}}}{
			\lIf{$state =$ {\upshape \texttt{new}}}{
				$state \gets \texttt{normal}$ \\
			}
			$myPhase \gets \mathsf{Random}([0,1,\dots,ll-1])$		    
		}{\textbf{endif}}\\
		\If{$phase = myPhase$}{
			$\GD(k)$}{\textbf{endif}}\\
	}
	
	\caption{$\GB$}
	
\end{algorithm}

Observe that the algorithm executed by a node that received the message consists of $3k = 72\lceil \log n\rceil$ \emph{phases}. Each phase consists of $t_{ph} = \BigO{\log n}$ rounds. In the analysis we group $ll = \lceil \log\log n \rceil$ consecutive phases into an \emph{epoch}. 

Let $T(v)$ denote the time when $v$ receives the message. Note that with a nonzero probability our algorithm may fail to deliver the message to some nodes. For all such nodes $v$ we define $T(v) = \infty$. Bounding the difference between $T(v)$ and $T(w)$ for any neighbors $v$ and $w$ is a key component in the analysis of the algorithm.
\begin{lemma}
	\label{lem:dom}
	If $n\geq 32$ then for any two neighbors $v,w \in V$ such that $T(v) < \infty$ and $T(v) \leq T(w)$ and for any $1 \leq x \leq 2\lceil\log{n}\rceil + 2$:
	\[
	\Pro{|T(v) - T(w)| \geq (x+2) \cdot t_{ph}} \leq 2^{-x}.
	\]
\end{lemma}
\begin{proof}
	Consider $(x+2) \cdot t_{ph}$ steps, starting from $T(v)$, during which $v$ executes algorithm \GB. It performs at least $x+1$ complete phases. From the definition of the algorithm we can observe that station $v$ executes $a < ll$ phases in state \texttt{new}, then it participates in $b$ full epochs and finally $c < ll$  phases of the last (incomplete) epoch. Hence $v$ participates altogether in the considered time interval in $b+2$ epochs. We will analyse the first epoch, consisting of $a$ phases separately.
	
	Observe that a station can finish its algorithm only at the end of a phase (see lines $6$ and $25$). Define by $\kappa_i$ the number of neighbors of $v$ participating in $i$-th epoch i.e. the number of stations that participate in the $i$-th epoch in all its phases in state \texttt{normal}. 
	
	\begin{enumerate}
		\item Epochs $\{2,3,\dots,b+2\}$. 
		\begin{enumerate}
			\item If $\kappa_i \geq 12 \lceil\log{n}\rceil \log\log{n}$ then the probability that one fixed of the $ll$ procedures $\GD$ is not executed by any station is at most
			\begin{align}\nonumber
			\left(1-\frac{1}{\lceil\log\log{n}\rceil}\right)^{12\lceil\log{n}\rceil\log\log{n}} & 
			\leq\left(\frac{\left(1 - \frac{1}{\lceil \log\log n\rceil}\right)^{\lceil \log\log n \rceil}}{1 - \frac{1}{\lceil \log\log n\rceil}}\right)^{12 \lceil \log n\rceil} \\\nonumber&
			\leq e^{-12 \cdot \lceil \log{n} \rceil } \cdot\left(1 + \frac{1}{\lceil \log\log n\rceil - 1}\right)^{12 \lceil \log n\rceil} \\ &\leq n^{-3},
			\label{eq:bb}
			\end{align}
			where  $12\lceil \log n \rceil \geq 17 \ln n$ and the inequality $\left(1 + \frac{1}{\lceil \log\log n\rceil - 1}\right)^{12 \lceil \log n\rceil} \leq n^{14}$ is true for $n \geq 32$. Hence with probability at least $n^{-2}$ (by the Union Bound) each $\GD$ has at least one participant. But then by Theorem~\ref{thm:gd} each $\GD$ it is successful with probability at least $1/2$.
			\item Fix any $i$ and assume that $\kappa_i < 12\lceil\log{n}\rceil\log\log{n}$. Observe that exactly $\kappa_i$ stations are taking part in $\BB$ procedure (in line 16) which is executed by all stations in \texttt{normal} state always in the first phase in each epoch. Now using Lemma~\ref{lem:bb} if $\kappa_i < 12\lceil\log{n}\rceil\log\log{n}$, the failure probability of $\BB$ is at most $1/(2\log{n}) \leq 2^{-ll}$.
		\end{enumerate}
		\item Epoch $1$. Denote by $\kappa$ the number of stations that execute procedure $\BB$ in state \texttt{new} together with $v$ (line $13$). Observe that if $\kappa < 12\lceil\log{n}\rceil\log\log{n}$ then the argument is the same as in the previous case which gives success with probability at least $1-2^{-ll}$. In the opposite case all these $\kappa$ stations choose a phase (line $21$) and by~\eqref{eq:bb} each of the $a$ procedures $\GD$ is executed by at least one station with probability at least $n^{-2}$ and hence successful with probability at least $1/2$.  
	\end{enumerate}
	
	We showed that if an epoch $i$ has $\kappa_i \geq 12 \lceil\log{n}\rceil \log\log{n}$ then the probability of success in each phase is at least $1/2 - n^{-2}$. And in the opposite case the probability of success of the entire epoch is
	at least 
	$1 - 2^{-ll}$. Let $y$ be the number of epochs $i$ with $\kappa_i \geq 12 \lceil\log{n}\rceil \log\log{n}$ and $z = x + 1 - y$. Then, using the independence of the phases, the failure probability is at most:
	\[ (1/2 + n^{-2})^{y\cdot ll} \cdot 2^{-ll \cdot z}  = 2^{-x-1} \cdot (1 + 2\cdot n^{-2})^y \leq 2^{-x},\]
	because $(1 + 2\cdot n^{-2})^{y\cdot ll} \leq 2$ holds already for $n> 10$ (because $y \leq \lceil \log{n}\rceil$ and $ll = \lceil \log\log{n}\rceil $). Now we note that in both cases the probability of success in each considered epoch is at least $2^{-i}$ where $i$ is the number of phases executed by $v$ in considered time interval $[T(v),T(v)+(x+2) \cdot t_{ph}]$. When we multiply the failure probabilities (the successes in epochs are independent) we obtain the failure probability of at most $2^{-x}$.
	
\end{proof}
In Lemma~\ref{lem:dom} we showed local bounds on the difference between the reception times of adjacent vertices. Using the Lemma we can show a global bound on the total number of step until all the nodes receive the message.
\begin{theorem}
	\label{thm:gd_broadcast}
	If $n\geq 32$ then algorithm \GB completes broadcast:
	\begin{enumerate}
		\InConference{\vspace*{-0.5mm}}
		\item in time $\BigO{(D+\log{n})\log{n}}$,
		\InConference{\vspace*{-0.5mm}}
		\item using $\BigO{\log{n} /\log\log{n}}$ energy per station,
		\InConference{\vspace*{-0.5mm}}
		\item with probability at least $1-2/n$.
	\end{enumerate}
\end{theorem}
\begin{proof}
	Similarly as in Section~\ref{sec:bb} we define a tree of shortest paths connecting $u$ to all the other nodes. For any vertex $v \in V$ we denote path $\mathcal{P}_v$ (in the tree) connecting $u$ (the originator) and $v$ and define a vertex $p(v)$ as the second last node on this path. The length of the path satisfies $D_v \leq  D$.  For any node $w$, the probability that some node $p(w)$ receives the message and $w$ does not is by Lemma~\ref{lem:dom} at most $1-1/n^2$ (since in the algorithm each station executes the outer loop $2\log n + 2$ times) and hence by the union bound all the nodes receive the message with probability at least $1-1/n$. Moreover at least one phase is successful with probability at least $1-1/n$ for any node $v$ at every vertex of the path $\mathcal{P}_v$. Call this event $\mathcal{E}$.
	
	We have:
	\begin{equation}
	\label{eqn:tsum}
	T(v) = T(v) - T(v_{D_v}) + \sum_{i=2}^{D_v} T(v_{i}) - T(v_{i-1}) + T(v_{1}).
	\end{equation}
	Thus let us denote $X_i = |T(v_{i}) - T(v_{i-1})|$, for $i=2,\dots,D_v$ and $X_{D_v+1} = |T(v) - T(v_{D_v})|$ and $X_1 = |T(v_{1})|$. Then by~\eqref{eqn:tsum} $T_{v} \leq \sum_{i=1}^{D_v+1} X_i$. Variables $X_i$ are not independent but their distribution is bounded in Lemma~\ref{lem:dom}. 
	Let us define a sequence of independent variables $Y_i \sim \mathsf{Geo(1/2)}$. Conditioned on event $\mathcal{E}$, each variable $X_i$ is stochastically dominated by $(Y_i + 2) \cdot t_{ph}$ by Lemma~\ref{lem:dom}. If $Y = \sum_{i=1}^{D + 2\log n} Y_i$ then by concentration bound for sum of geometric variables~\cite[Theorem 2.3]{janson2017tail} we have:
	\[
	\Pro{Y \geq 4 \cdot \E{Y}} \leq n^{-2}.
	\]
	Hence again by the union bound this holds for all nodes with probability at least $1-1/n$. Thus with probability at least $1-2/n$ event $\mathcal{E}$ takes place and $Y < 4 \cdot \E{Y}$ for all vertices $v$. Hence with probability $1-2/n$ all stations receive the message within time $(D + 2\log n)\cdot t_{ph}$. All stations terminate the algorithm after a most $2\log{n} + 2$ additional phases hence the total time is $\BigO{(D + \log{n})\log{n}}$, which completes the proof of $1$ and $3$. 
	
	The energy complexity follows directly from the fact that the energy used by each station in $\log\log{n}$ consecutive phases is always constant. Hence the energy is at most $\BigO{\log{n}/\log\log n}$
\end{proof}	
\section{Lower bound}\label{lower2}
Our algorithm $\GGB$ has a surprisingly large multiplicative factor $n^{1/\varphi}$. In this section we want to prove that such a factor is sometimes necessary by showing that time $\Omega(n^{1/\varphi}\cdot \varphi)$ is needed for any algorithm using energy $\varphi$. 
\begin{theorem}
	\label{thm:lower}
	For any randomized broadcast algorithm $\mathcal{A}$ successful with probability at least $(1-e^{-1})/2 $ in all multi-hop radio networks there exists a graph $G$ with constant diameter and $n$ nodes, such that if $T$ is the runtime and $E$ is the energy used by the algorithm then the expected value of $E\cdot \log{(T/E)}$ is $\Omega(\log n)$.
\end{theorem}
{
	\begin{proof}
		We want to prove the theorem using Yao's minimax principle~\cite{yao}. Take any deterministic algorithm $\mathcal{A}$ and consider any fixed $n$. We define a family of graphs $\mathcal{G}$ on $n$ vertices. For simplicity assume that $n-1$ is divisible by $2$. For any $G_{\pi}\in\mathcal{G}$ we have $G_{\pi} = (\{u\}\cup S \cup X, E_s \cup E_{\pi})$, where $|S| = |X| = \frac12(n-1)$. Node $u$ is the originator of the message in each $G_{\pi}\in\mathcal{G}$. Also in each $G_{\pi} \in \mathcal{G}$, set $E_s$ is defined as $E_s = \{(u,s) : s\in S\}$ (i.e., vertices $\{u\} \cup S$ are forming a star with $u$ at its center). Graphs from $\mathcal{G}$ differ on the remaining edges from $E_{\pi}$ in the following way. Let $S = \{s_1,s_2,\dots,s_{(n-1)/2}\}$ and $X = \{x_1,x_2,\dots,x_{(n-1)/2}\}$. Take the set $\Pi$ of all permutations $\pi: [(n-1)/2] \rightarrow [(n-1)/2]$, such that $\pi(i) \neq i$, for any $i\in [(n-1)/2]$. For all $\pi \in \Pi$ we define $E_{\pi}$ as $\{(s_i,x_i),(x_i,s_{\pi(i)}) : i\in \{1,2,\dots,(n-1)/2\}\}$. Now consider algorithm $\mathcal{A}$ on graph $G_{\pi}$ taken uniformly at random from $\mathcal{G}$. Let $T$ denote the time of the algorithm and $E$ denote maximum energy used by the stations. Consider the total number of possible broadcasting patterns of length $T$ with at most $E$ transmissions:
		\[
		\alpha(T,E) = \sum_{i = 0}^{E} {T\choose i} \leq \sum_{i=0}^E \frac{T^i}{i!} = \sum_{i=0}^E \frac{E^i}{i!}\cdot \left(\frac{T}{E}\right)^i \leq \left(\frac{T}{E}\right)^E \sum_{i=0}^\infty \frac{E^i}{i!} = \left(\frac{eT}{E}\right)^E.
		\]
		
		Now, assume that in algorithm $\mathcal{A}$, the last expression is upper bounded by $n/100$. Then also $\alpha(t,x) \leq  n/100$.  
		Assume that $n/100$ is an integer. Since there are at most $n/100$ broadcasting patterns there are then at least $\frac12 (n-1) - \frac{9n}{100}  \geq \frac13 n $ 
		stations from $S$ which have the same pattern as at least $10$ other stations from $S$. Call the set of these $\frac13 n$ stations $\hat{S}$. Observe by the construction of the graph that if the stations with the same patterns also have a common neighbor in $x\in X$ then they cannot break the symmetry because they receive the same feedback from the channel in each time step. And then the message cannot be delivered to this neighbor $x$. We want to lower bound the probability that two stations from $\hat{S}$ have a common neighbor in $X$. Set $\mathcal{G}$ is defined in such a way that in a graph $G_{\pi}\in\mathcal{G}$ chosen uniformly at random, for each $s \in S$, its corresponding $\pi(s)$ can be seen as a vertex taken uniformly from $S\setminus\{s\}$. Thus, for any $s \in \hat{S}$, the probability that $p(s)$ is using different broadcasting pattern is at most $1- 5/(n-1)$. Hence with probability at most,
		$\left(1-\frac{5}{n-1}\right)^{n/3} \leq 1/e,$
		for each station $s\in S$ its corresponding $p(s)$ station uses a different pattern. Hence under the chosen probability distribution over the set of graphs, with  probability at least $1-e^{-1}$, some node does not receive the message. Hence if we define as the cost of the algorithm the expression $(eT/E)^E$ then its expected value is $\Omega(n)$. By Yao's principle~\cite[Theorem 3]{yao} for Monte Carlo algorithms, for any randomized algorithm with error probability at most $(1-e^{-1})/2$ there exists graph $G_{\pi}\in\mathcal{G}$ such that the expected value of $(eT/E)^E$ is $\Omega(n)$.

		
		
		

	\end{proof}
}
The following Corollary lower bounds the time complexity of any algorithm using the asymptotically same energy as the algorithms presented in Section~\ref{green2}.
\begin{corollary}
	\label{cor:lower}
	Any randomized algorithm completing broadcast in any graph with probability at least $(1-e^{-1})/2$
	\begin{enumerate}
		\item using energy at most $\varphi$ needs expected time $\Omega(n^{1/\varphi}\cdot \varphi)$,
		\item using energy at most $\frac{\log{n}}{c\log\log{n}}$ for any constant $c$ needs expected time $\Omega\left(\frac{\log^{c+1}n}{\log\log n}\right)$.
	\end{enumerate}
\end{corollary}
This shows that for graphs with constant diameter our algorithms achieve almost optimal tradeoff between time and energy. 
This also shows that our $\GB$ is asymptotically optimal in terms of energetic efficiency among all algorithms with time polylogarithmic in $n$.

\begin{fact}
	\label{fct:collision}
	For any discrete positive distribution $\mathcal{D}$ and two independent variables $X,Y \sim \mathcal{D}$ if $\E{X} = \E{Y}  = k$, for even $k \geq  1$ then $\Pro{X = Y} \geq 1/(2k)$.
\end{fact}
\begin{proof}
	Denote $\Pro{X = i} = p_i$. We know by Markov inequality that $\sum_{i = 1}^{k/2} p_i \geq 1/2$. By Cauchy-Schwarz inequality:
	\[
	\left(\sum_{i=1}^{k/2} p_i\right)^2 \leq \left(\sum_{i=1}^{k/2} p_i^2\right)\left(\sum_{i=1}^{k/2} 1 \right)
	\] 	
	\[
	\frac{1}{4} \leq \left(\sum_{i=1}^{k/2} p_i^2\right)\frac{k}{2} = \frac{k}{2}\sum_{i=1}^{k/2}\Pro{X = i, Y = i}
	\] 	
	
\end{proof}
\begin{figure}
	\centering
	\includegraphics[width=\linewidth]{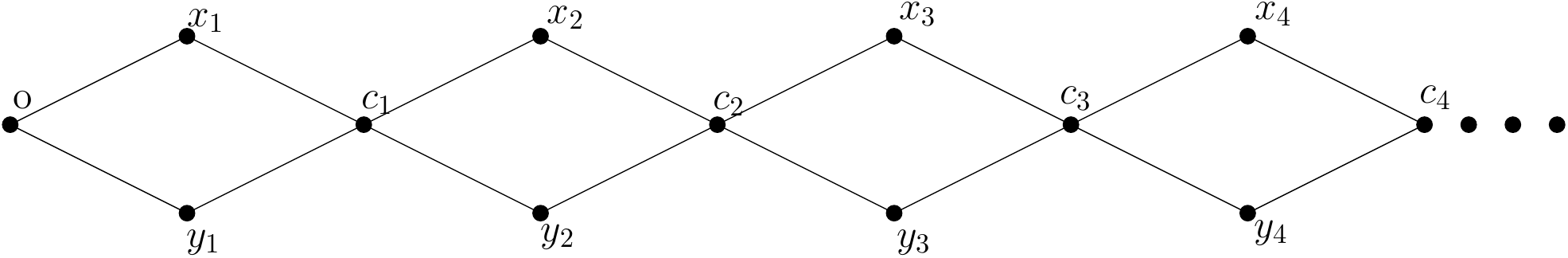}
	\caption{Graph used in proof of Theorem~\ref{thm:energy2}.}
	\label{fig:lancuch}
\end{figure}
\begin{theorem}
	\label{thm:energy2}
	Any randomized algorithm completing broadcast in any graph with probability at least $1 - 1/n$ using at most $2$ energy needs expected time $\Omega(n^{3/2})$ on some graphs with diameter $D = \Theta(n)$.
\end{theorem}
\begin{proof}
	In the proof we use the ``chain'' graph $C_{n}$ (see Figure~\ref{fig:lancuch}) of diameter $n/3$ (assume for simplicity that $n$ is divisible by $24$ and that $\sqrt{n}$ is a natural number divisible by $4$). Node $o$ is the originator of the message in the broadcast problem. 
	
	Take an algorithm $\mathcal{A}$ that solves global broadcast with probability at least $1 - 1/n$ on any graph and consider its execution on graph $C_{n}$. Take any pair $x, y$ in the graph. The stations do not have identifiers and cannot transmit before receiving the message hence their initial position is symmetric at the moment of receiving the message for the first time. We can denote $T_x^{(1)}$ and $T_y^{(1)}$ as the time of the first transmission of $x$ and $y$ respectively (time measured as the number of steps after receiving the message). First observe that if $\E{T_{x}^{(1)}} \geq \sqrt{n}/2$ then this means that every node on average waits $\sqrt{n}/2$ steps until it transmits for the first time after receiving the message. Hence $\sum_{i = 1}^{n/3} \E{T_{c_i}} \geq n^{3/2} / 6$ and the delay caused by the $c$-nodes is $\Omega(n^{3/2})$. Thus assume that $\E{T_{x}^{(1)}} \leq \sqrt{n}/2$. Then by Fact~\ref{fct:collision} nodes $x_i$ and $y_i$ transmit for the first time in the same slot with probability at least $1/\sqrt{n}$. Thus by Chernoff bound if $n \geq 64$ then with probability at least $1/2$ we have at least $\sqrt{n}/2$ pairs $x,y$ that conflicted in the first transmission. We call such pairs \emph{conflicting pairs}.
	
	Since the algorithm works with probability at least $1-1/n$ hence the conflict between all such conflicting pairs has to be resolved with high probability using the second transmission as otherwise the algorithm fails to deliver the message to some $c$ node.
	
	Observe that if $x$ and $y$ made the first transmission in the same first slot then they are still in a symmetric position and the distribution of their second transmission is also the same. Let random variables $T_x^{(2)}$ and $T_y^{(2)}$ denote the times of the second transmission (measured as the number of steps after the first one).  Assume that $\Pro{T_x^{(2)}\leq n/4} \geq  1/2$ and denote by event $L$ the event that $T_x^{(2)}\leq n/4$ and $T_y^{(2)}\leq n/4$. Note that by Fact~\ref{fct:collision} $\Pro{T_x^{(2)} = T_y^{(2)} | L} \geq 4/n$ and the probability of failure (not delivering the message to node $c$) in this case is at least $4/(n^{3/2})$ (because collision in the second slot is independent from the collision in the first slot). This means that we can have at most $\sqrt{n}/4$ such conflicting pairs as otherwise it would violate the guarantee of the algorithm. Hence for the remaining $\sqrt{n}/4$ conflicting pairs $x,y$ we have $\Pro{T_x^{(2)}\leq n/4} \leq  1/2$ but then by the independence of $T_x^{(2)}$ and $T_y^{(2)}$ we have that $\Pro{T_x^{(2)} \geq n/4, T_y^{(2)} \geq n/4} \geq 1/4$. This shows that the number of steps between the reception of the message by $x$ and $y$ and the reception by corresponding $c$-node is $\Omega(n)$ with probability at least $1/4$ for at least $\sqrt{n}/4$ conflicting pairs. Thus the total expected time of the algorithm is $\Omega(n^{3/2})$.
\end{proof}
Our lower bound shows separation between the models with known and unknown topology. If the topology is known then it is possible to solve broadcast in time linear in $D$ whereas for unknown topology $\Omega(n^{3/2})$ is required for energy $2$ and diameter $\Theta(n)$. Hence in the unknown topology model it is impossible to design an algorithm working in time $O(D + n^{1/\varphi} \varphi)$ in the full range of values of $\varphi$. An interesting future direction is to generalize the lower bound for any $\varphi$ and obtain a bound of the form $\Omega(D n^{1/\varphi})$.
\bibliographystyle{abbrv}
\bibliography{bibliography} 

\end{document}